\documentclass[12pt]{article}

\usepackage{amsmath}
\usepackage{amsfonts,amssymb,amsthm,overpic}
\makeatletter
\newcommand{\diam}{\mathop{\operator@font diam}}
\usepackage{setspace}
\usepackage{multicol}
\usepackage{multirow}
\usepackage[compact]{titlesec}
\usepackage{xcolor}

\usepackage{epsfig}

\newtheorem{proposition}{Proposition}[section]

\newtheorem{theorem}{Theorem}[section]

\newtheorem{corollary}{Corollary}[section]

\begin{document}

\title{\Huge{\textsc{On Completeness of the Alexandrov Topology on a Spacetime: some remarks and questions}}}

\author{Kyriakos Papadopoulos$^1$, Nazli Kurt$^2$\\
\small{1. Department of Mathematics, Kuwait University, PO Box 5969, Safat 13060, Kuwait}\\
\small{2. Open University, UK}\\
E-mail: \textrm{ kyriakos.papadopoulos1981@gmail.com}
}

\date{}

\maketitle

\begin{abstract}
We clarify and discuss a misunderstanding between uniform completeness and metric completeness, that has appeared in the literature in a study on the Alexandrov topology for a spacetime.
\end{abstract}

\section{Completeness with respect to Alexandrov Topology for Spacetime: revisited.}

It will be far beyond the scope of this note to present analytically all the topological background that one needs to deal with a terminology including words like topological space, compactness, precompact set, metric space, metric-completeness, uniform space, uniform completeness, paracompactness, Hausdorff space, Tychonoff space; for the definition of these terms, we refer to the classical book of Engelking \cite{Topology}. In addition, for the definitions of spacetime, the Alexandrov topology on a spacetime, causal-diamonds, strong causality and global hyperbolicity, the classical reading of Penrose, \cite{Penrose-difftopology}, will be  more than sufficient.

It seems that the study of the Alexandrov topology, on a spacetime, is underestimated, and is limited to the following result of Theorem 4.24, of \cite{Penrose-difftopology}:

\begin{theorem}
Let $M$ be a spacetime. Then, the following are equivalent.
\begin{enumerate}
\item $M$ is strongly causal.
\item The Alexandrov topology agrees with the manifold topology.
\item The Alexandrov topology is Hausdorff.
\end{enumerate}
\end{theorem}

An attempt to link the Alexandrov topology to metric completeness of a spacetime appeared in \cite{McW}, with an intention to examine possible implications in geodesic completeness. The authors of \cite{KP} used this result in order to show that in a sliced spacetime, where the product topology coincides with the Alexandrov topology, the completeness of a slice with respect to Alexandrov topology is equivalent to the spacetime being globally hyperbolic.

 It is unfortunate  that there is a confusion in the literature between the terms uniform completeness and metric completeness, which is due to the occasionally unlucky use of the term “topological completeness”, for both notions. In fact, metric spaces are uniform (and metrizable spaces are uniformizable complete), but not vice-versa. The purpose of this note is to expose this confusion in the powerful result of \cite{McW} and to set a frame for a discussion on the use of uniform completeness in the study of the convergence of causal curves.

We denote by $T_A$ the Alexandrov topology on a spacetime $(M,g)$ and by $T_A(X)$ the induced topology on a subset $X$, of $M$. \footnote{It is reasonable to consider subsets $X$, in which the spacetime metric $g$ is nondegenerate at each point.} Below we state the theorem which is the main motivation for our discussion.

\begin{theorem}[McWilliams]\label{1}
Let $(M,g)$ be a spacetime. Then, for any subset $X \subset M$, $T_A(X)$ is complete, if and only if $X$ is strongly causal. 
\end{theorem}

Given the $\Rightarrow$ direction of the proof, it is certain that when the author mentions ``completeness'', he refers to ``metric completeness'', as the proof states ``If $T_A(X)$ is complete, then by definition $T_A(X)$ is metric, hence Hausdorff, hence strongly causal''. 

If the intention was to talk about topological (in other words uniform) completeness, the Hausdorfness of $T_A(X)$ would not be guaranteed.
 

We now focus on 
the $\Leftarrow$ direction of the proof. That the strong causality of $X$ implies that $T_A(X)$ is Tychonoff ($T_1$-completely regular) is trivial, and comes from the fact that in the case of $X$ being strongly causal, $T_A(X) = T_M(X)$, where $T_M$ is the manifold topology on $(M,g)$ (see \cite{Penrose-difftopology}, p. 34). So, since every subspace of a metric space is metric and since metric spaces are Tychonoff, the first result follows naturally. Now, Tychonoff spaces are uniformizable (see \cite{Topology}) and since $T_M(X)$ is paracompact (as a metric space), $T_A(X)$ will be paracompact, too. Paracompact uniformizable spaces, though, are ``topologically complete'', according to W. Page, in \cite{Page} (p. 54), where by topological completeness it is meant uniformizable completeness; i.e. there is a complete uniformity $\mathcal{U}$ which is compatible with the topology $T_A(X)$. Since not every uniform space is metrizable (see \cite{Topology}), the $\Leftarrow$ direction of Theorem \ref{1} is inconsistent with the proof of the $\Rightarrow$ direction. 
\\

{{\bf Remarks:}} 

\begin{enumerate}

\item The $\Leftarrow$ direction of the proof is actually trivial, if the intention is to show metric completeness.
Since $X$ is strongly causal, the Alexandrov topology $T_A(X)$ on $X$ agrees with the manifold topology $T_M(X)$ on $X$. Since also the manifold topology is metric complete (naturally, one considers a spacetime to be second countable) and every open subset of a manifold is a manifold itself, $T_M(X)$ will be metric complete, too. Metric completeness implies uniform completeness, but not vise versa; that is why the proof of the $\Leftarrow$ direction of Therem \ref{1} is inconsistent.

\item If the intention of $\Leftarrow$ is to show uniform completeness, indeed, then one would have to introduce conditions, so that uniform completeness to imply strong causality, in the $\Rightarrow$ direction (in 1. we have already shown that strong causality implies uniform completeness). In other words, one would have to introduce conditions, so that $X\subset M$ being uniform complete implies that $T_M(X)$ coincides with $T_A(X)$.

\item In the unusual (but not to be neglected) case that the spacetime manifold $M$ is not considered to be second countable,  there will be no globally defined derivative operator; instead, the manifold structure will admit derivative operators locally. In this particular case, the manifold topology would not be metric complete and uniform completeness might play a role, instead. In general, a space being metric complete does not imply that every subset with the subspace topology will be metric complete, too. It is a well-known result that if $X$ is a closed subset of a metric complete space, then $X$ is complete, too, but metric completeness is not inherited by all subsets. Furthermore, if one abstracts an event set equipped with the usual chronological and causal order, from a spacetime manifold (see \cite{Kron}), one could associate with the event set a uniform structure, and study uniform convergence, completeness and continuity, abstractly and far and away from the metric structure of a manifold. Since being uniform complete is a weaker condition than being metric complete, it might be that uniform completeness will play a role in the study of the global causality properties of a spacetime and, in particular, the convergence of causal curves being studied in the frame of filters, rather than sequences of points.  So, a theorem, connecting the causal structure of spacetime with uniform completeness will be desirable.

\end{enumerate}

{{\bf Some Further Remarks (on reasonable spacetimes):}}

Here we adopt the definition of Hounnonkpe-Minguzzi for a reasonable spacetime: a spacetime is reasonable, if it is noncompact and of dimension strictly greater than $2$ (see \cite{Minguzzi}).

Let $\mathcal{C}$ denote the space of smooth endless causal curves, equipped with topology $T^0$, defined  via basic-open sets as follows (see \cite{Low}); consider  a smooth endless causal curve $\Gamma$ and a point $x \in \Gamma$. Consider an open-neighbourhood $N$ of $x$ in the manifold topology $T_M$, and let $U$ be the set of all smooth endless causal curves passing through  $N$. Then, $U$ is a basic-open set of the topology $T_0$, which is the natural topology on $\mathcal{C}$ encapsulating pointwise convergence as well as convergence of tangent directions.

The following two propositions talk about global hyperbolicity against topological properties of the space $\mathcal{C}$ (see \cite{Low}).

\begin{proposition}[Low]\label{L1}
A strongly causal spacetime $M$ is globally hyperbolic, if and ony if $\mathcal{C}$ is Hausdorff.
\end{proposition}

\begin{proposition}[Low]\label{L2}
$M$ is globally hyperbolic, if and only if $\mathcal{C}$ is metrizable.
\end{proposition}

In particular, for the proof of Proposition \ref{L2}, Low constructs a metric that induces the topology $T^0$.

In order to talk about a connection between the (seemingly different) topologies $T_A$ and $T^0$, for reasonable spacetimes, we will need to use the following result (see \cite{Minguzzi}).

\begin{theorem}[Hounnonkpe-Minguzzi]\label{Minguzzi}
 A non-compact spacetime of dimension strictly greater than $2$ is globally hyperbolic, if and only if the causal diamonds are compact in the Alexandrov topology, $T_A$.
 \end{theorem}

The next result then follows naturally.

\begin{corollary}\label{3}
Let $(M,g)$ be  a non-compact spacetime of dimension $\geq 3$. Then, if the basic-open sets of the Alexandrov topology are precompact, the topology $T^0$ is uniformizably complete.
\end{corollary}
\begin{proof}
Precompactness of the Alexandrov basic-open sets gives strong causality (see Theorem \ref{Minguzzi}) and global hyperbolicity in $M$ is equivalent to the metrizability of $T^0$. The metrizability of $T^0$ implies its uniformizable completeness.
\end{proof}

Corollary \ref{3} raises the following questions; under what conditions will the implication be an equivalence? In other words, under what conditions will the uniformizable completeness of $T^0$ imply the precompactness of the Alexandrov basic-open sets? More interestingly: are there natural conditions which give an equivalence between uniform completeness of the Alexandrov topology and uniform completeness of $T^0$? Answers to these questions might shed a light to the question: under what conditions, in a spacetime,  is geodesic completeness equivalent to uniformizable completeness or, in particular, metric completeness? It is evident that the naive analogue of Hopf-Rinow theorem fails to exist in Lorentzian manifolds, and since it is more reasonable to consider global hyperbolicity as a condition that two events, that are chronologically related, can be joined by a maximal timelike geodesic, it might be that Corollary \ref{3} gives some hints about the way that one should follow towards this discussion. 

\section{Summazing Questions.}
The inconcistency that we showed that appears in the proof of Theorem \ref{1} inspires us to consider the following. The role of uniform completeness in spacetime and how it differs from metric completeness could be of interest; in particular, the space $\mathcal{C}$ of endless causal curves is not metric complete in its natural topology, in general; how one should understand this? Are there singular missing points in the discussion found in \cite{McW}? More interestingly, what about the obstruction to topology change argument, presented by Geroch \cite{Geroch}; does it, in fact, resolve the contradiction?

\end{document}